\documentclass[letterpaper,11pt]{article}
 
\usepackage{microtype}

\title{Further limitations of the known approaches for matrix multiplication}
\author{Josh Alman\footnote{MIT CSAIL and EECS, jalman@mit.edu. Partially supported by an NSF Graduate Research Fellowship.} \and Virginia Vassilevska Williams\footnote{MIT CSAIL and EECS, virgi@mit.edu. Partially supported by an NSF Career Award, a Sloan Fellowship, NSF Grants CCF-1417238, CCF-1528078 and CCF-1514339, and BSF Grant BSF:2012338.}}

\usepackage[utf8]{inputenc}
\usepackage{amsmath,amsfonts,amssymb,amsthm,fullpage}
\usepackage{graphicx}
\usepackage{caption}
\usepackage{subcaption}

\newtheorem{theorem}{Theorem}[section]

\newtheorem{corollary}{Corollary}[section]
\newtheorem{lemma}{Lemma}[section]

\def \Z {{\mathbb Z}}
\def \F {{\mathbb F}}
\def \R {{\mathbb R}}
\def \E {{\mathbb E}}
\def \T {{\mathbb T}}

\def \i {{\vec i}}
\def \j {{\vec j}}
\def \k {{\vec k}}

\def\eps{\varepsilon}

\newcommand{\Mod}[1]{\ (\mathrm{mod}\ #1)}

\begin{document}

\maketitle

\begin{abstract}
We consider the techniques behind the current best algorithms for matrix multiplication. Our results are threefold. 

(1) We provide a unifying framework, showing that all known matrix multiplication running times since 1986 can be achieved from a single very natural tensor - the structural tensor $T_q$ of addition modulo an integer $q$. 

(2) We show that if one applies a generalization of the known techniques (arbitrary zeroing out of tensor powers to obtain independent matrix products in order to use the asymptotic sum inequality of Sch\"{o}nhage) to an arbitrary monomial degeneration of $T_q$, then there is an explicit lower bound, depending on $q$, on the bound on the matrix multiplication exponent $\omega$ that one can achieve. We also show upper bounds on the value $\alpha$ that one can achieve, where $\alpha$ is such that $n\times n^\alpha \times n$ matrix multiplication can be computed in $n^{2+o(1)}$ time.

(3) We show that our lower bound on $\omega$ approaches $2$ as $q$ goes to infinity. This suggests a promising approach to improving the bound on $\omega$: for variable $q$, find a monomial degeneration of $T_q$ which, using the known techniques, produces an upper bound on $\omega$ as a function of $q$. Then, take $q$ to infinity. It is not ruled out, and hence possible, that one can obtain $\omega=2$ in this way.
 \end{abstract}

\thispagestyle{empty}
\newpage
\setcounter{page}{1}

\section{Introduction}

One of the most fundamental questions in computer science asks how quickly one can multiply two matrices. Since the surprising subcubic algorithm for $n \times n \times n$ matrix multiplication by Strassen in 1969~\cite{strassen}, there has been a long line of work on improving and refining the techniques and speeding up matrix multiplication algorithms (e.g.~\cite{Pan78,Pan80,BCRL79,laser,cw81as,sch81,strassenlaser1,coppersmith,stothers,v12,legall}). Progress on this problem is typically measured in terms of $\omega$, the smallest constant such that, for any $\delta>0$, one can design an algorithm for $n \times n \times n$ matrix multiplication running in time $O(n^{\omega + \delta})$. The biggest open question is whether one can achieve $\omega=2$. The best bound we currently know, due to Le Gall~\cite{legall}, is $\omega \leq 2.3728639$. 

A related line of work \cite{coppersmith, coppersmith1997rectangular, legallrect, legallrect2} focuses on \emph{rectangular} matrix multiplication instead of square matrix multiplication. Here, progress is measured in terms of $\alpha$, the largest constant such that for any $\delta>0$, one can design an algorithm for $n \times n^\alpha \times n$ matrix multiplication running in time $O(n^{2 + \delta})$. Recent work \cite{legallrect2} improved the best known bound to $\alpha > 0.31389$. The two values $\omega$ and $\alpha$ are very related, as $\omega=2$ if and only if $\alpha=1$.

All of the aforementioned bounds on $\omega$ and $\alpha$ follow a particular approach, which works as follows.\footnote{We give a very high level overview here. More precise definitions are given in Section \ref{sec:prelims}. For a more gentle introduction, we recommend the notes by Markus Bl{\"a}ser \cite{blaser}.} The key is to cleverly select a trilinear form (third-order tensor) $\T$ which needs to have two properties. First, there must be an efficient way to compute large tensor powers $\T^{\otimes n}$ of $\T$. This is done by finding a low \emph{border rank expression} for $\T$, which implies (via Sch{\"o}nhage’s asymptotic sum inequality) that for sufficiently large $n$, the power $\T^{\otimes n}$ has low rank. Second, $\T$ must be useful for actually performing matrix multiplication. Multiplying matrices corresponds in a precise way to evaluating a certain \emph{matrix multiplication tensor}, and so to use $\T$ for this task, one needs to show that there is a `degeneration' transforming $\T$ into a disjoint sum of matrix multiplication tensors. Combining these two properties of $\T$ yields an algorithm for matrix multiplication (see Lemma~\ref{lem:schon} below for the precise formula).

Of course, the resulting runtime depends on the choice of the tensor $\T$ as well as the bounds one can prove for the two desired properties. Strassen's original algorithm picked $\T$ to be the tensor for $2 \times 2 \times 2$ matrix multiplication itself. Later work used more and more elaborate tensors and corresponding border rank expressions, culminating with the most recent algorithms using the now-famous \emph{Coppersmith-Winograd tensor}. All these tensors seem to come `out of nowhere', and in particular, come up with seemingly `magical' border rank identities to show that they have low border rank. We make some progress demystifying the tensors and their border rank expressions below.

\subsection{The best known bounds on $\omega$ are actually from $T_q$.}
Our first result is a \emph{unifying approach} to achieving all known bounds of $\omega$ (\cite{laser,coppersmith,stothers,legall}) since Strassen's 1986 proof that $\omega<2.48$. 

A simple remark first pointed out to us by Michalek~\cite{michalek-personal} is that the so called Coppersmith-Winograd tensor used in the papers on matrix multiplication since 1990~\cite{coppersmith,stothers,legall}, can be replaced with an equivalent tensor, rotating the original slightly in a certain way (see the Preliminaries), without changing any of the proofs, and thus yielding the same bounds on $\omega$.

 With this in mind, we consider a tensor $T_q$, the \emph{structural tensor of $\Z_q$}, and give a very simple low rank expression for it based on roots of unity (this expression is natural and likely well-known). We then show that the tensor in \cite{laser} and the rotated Coppersmith-Winograd tensors that can be used in \cite{coppersmith,stothers,legall,v12},
  are all actually
 straightforward monomial degenerations of $T_q$. Since a monomial degeneration of a rank expression gives a border rank expression, this (for example) yields a straightforward border rank expression for the (rotated) Coppersmith-Winograd tensor, which is more intuitive than the border rank expressions from past work. 
 
 Another way to view this fact is that \emph{all the bounds on $\omega$ since \cite{coppersmith} can be viewed as using $T_q$ (in fact for $q=7$ or $8$) as the underlying tensor $\T$}! This also suggests a potential way to improve the known bounds on $\omega$: study other monomial degenerations of $T_q$.

\subsection{Limitations on monomial degenerations of $T_p$.}
Our second and main result is a lower bound on how fast a matrix multiplication algorithm designed in this way can be whenever $\T$ is a monomial degeneration of $T_p$:

\begin{theorem}[Informal] \label{thm:maininformal}
For every $p$, and for every $\eps \in (0,1]$, there is an explicit constant $\nu_{p,\eps} > 1$ such that any algorithm for $n \times n^\eps \times n$ matrix multiplication designed in the above way using $T_p$, or a monomial degeneration of $T_p$, runs in time $\Omega(n^{(1+\eps) \nu_{p,\eps}})$. (See Theorem~\ref{thm:main} below for the precise statement).
\end{theorem}

The constant $\nu_{p,\eps}$ is defined as follows. Consider first when $p$ is a fixed prime or power of a prime. Let $z$ be the unique real number in $(0,1)$ such that $3\sum_{j=1}^{p-1}z^j = (p-1) (1-2z^p)$; then 
$$\nu_{p,\eps}:=(1+\eps)\ln \left[\frac{1-z^p}{(1-z) z^{(p-1)/3}}\right].$$

There is also a variant of Theorem~\ref{thm:maininformal} that holds for $T_p$ when $p$ is not necessarily a prime power, but the constant $\nu_{p,\eps} > 1$ is slightly different.

In particular, this shows that:
\begin{itemize}
    \item This approach yields a square matrix multiplication algorithm with runtime at best $\Omega(n^{2 \nu_{p,1}})$, with exponent $2 \nu_{p,1} > 2$. Hence, this approach for a fixed $p$ cannot yield $\omega=2$.
    \item  Let $\eps_p \in (0,1)$ be such that $(1+\eps_p) \nu_{p,\eps} = 2$. Then, this approach for a fixed $p$ cannot yield a value of $\alpha$ bigger than $\eps_p$.
\end{itemize}

For modest values of $p$, the value $\nu_p:=\nu_{p,1}$ is a fair bit larger than $1$. For instance, $\nu_7 \approx 1.07065$. As we will show shortly, the best known algorithms for matrix multiplications use the approach above with a (rotated) Coppersmith-Winograd tensor which is a monomial degeneration of $T_7$. 
Our theorem implies among other things that using the approach with $T_7$ as the starting tensor cannot yield a bound on $\omega$ better than $2.14$, no matter how one zeroes out the tensor powers of $T_7$ or its monomial degenerations. 
We plot the resulting bounds on $\omega$ and $\alpha$ for varying $p$, in Figures~\ref{fig:graphs} and~\ref{fig:graphs2} (for technical reasons we discuss below, we get different bounds depending on whether $q$ is a power of a prime).

\begin{figure}[h!]
\centering
\begin{subfigure}{.5\textwidth}
  \centering
  \includegraphics[width=0.95\linewidth]{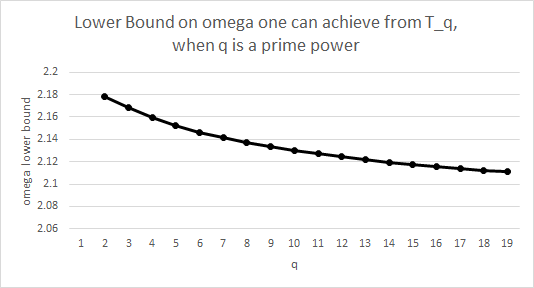}
  \caption{Lower bound on $\omega$ that one can achieve\\ using $T_q$ when $q$ is a power of a prime.\\ The bound approaches $2$ as $q \to \infty$.}
\end{subfigure}%
\begin{subfigure}{.5\textwidth}
  \centering
  \includegraphics[width=0.95\linewidth]{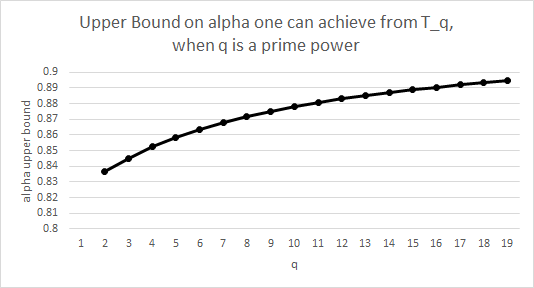}
  \caption{Upper bound on $\alpha$ that one can achieve\\ using $T_q$ when $q$ is a power of a prime.\\ The bound approaches $1$ as $q \to \infty$.}
\end{subfigure}
\caption{Bounds on $\omega$ and $\alpha$ that follow from Theorem~\ref{thm:maininformal} when $q$ is a prime power}
\label{fig:graphs}
\end{figure}

\begin{figure}[h!]
\centering
\begin{subfigure}{.5\textwidth}
  \centering
  \includegraphics[width=0.95\linewidth]{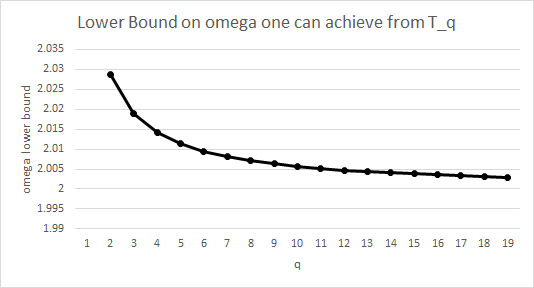}
  \caption{Lower bound on $\omega$ that one can achieve\\ using $T_q$.\\ The bound approaches $2$ as $q \to \infty$.}
\end{subfigure}%
\begin{subfigure}{.5\textwidth}
  \centering
  \includegraphics[width=0.95\linewidth]{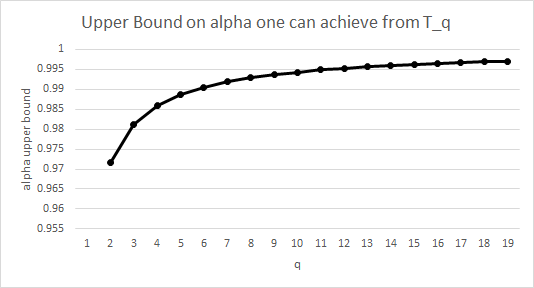}
  \caption{Upper bound on $\alpha$ that one can achieve\\ using $T_q$.\\ The bound approaches $1$ as $q \to \infty$.}
\end{subfigure}
\caption{Bounds on $\omega$ and $\alpha$ that follow from Theorem~\ref{thm:maininformal} for any $q$}
\label{fig:graphs2}
\end{figure}

\subsection{A potential idea for improving $\omega$.}
It should be noted that, despite our lower bounds, not all hope is lost for achieving $\omega=2$ using $T_q$ tensors. Indeed, in the limit as $q \to \infty$, our $\omega$ lower bound approaches $2$, and our $\alpha$ upper bound approaches $1$ (see Lemma~\ref{lem:goesto2} in Appendix~\ref{sec:calculations} for a proof). Hence, our lower bound does not rule out achieving a runtime for $n \times n \times n$ matrix multiplication of $O(n^{2 + \delta})$ for all $\delta>0$ by using bigger and bigger values of $q$. We find this approach very exciting.

\subsection{Tri-Colored Sum-Free Sets} \label{sec:trisets}

A key component of our lower bound proof is a recent upper bound proved on the asymptotic size of a family of combinatorial objects called tri-colored sum-free sets. For an abelian group $G$, a \emph{tri-colored sum-free set in $G^n$} is a set of triples $(a_i, b_i, c_i) \in (G^n)^3$ such that $a_i + b_j + c_k = 0$ if and only if $i=j=k$. In this paper we are especially interested in tri-colored sum-free sets over $\Z_q^n$.

Recent work \cite{eg,kleinberg,blasiak,norin2016,pebody} has proved upper bounds on how large tri-colored sum-free sets in $\Z_q^n$ can be. The bound is originally given in terms of the entropy of certain symmetric distributions, but we give a more explicit form written out by \cite{norin2016,pebody} here.

For any integer $q \geq 2$ which is a power of a prime, let $\rho$ be the unique number in $(0,1)$ satisfying
$$\rho + \rho^2 + \cdots + \rho^{q-1} = \frac{q-1}{3} (1 + 2 \rho^q ).$$
Then, define $\gamma_q \in \R$ by $\gamma_q := \ln(1-\rho^q) - \ln(1-\rho) - \frac{q-1}{3} \ln(\rho)$.

\begin{theorem}[\cite{kleinberg}]
Let $q$ be any prime or power of a prime. Then, any tri-colored sum-free set in $\Z_q^n$ has size at most $e^{\gamma_q n}$. Moreover, there exists a tri-colored sum-free set in $\Z_q^n$ with size $e^{\gamma_q n - o(n)}$.
\end{theorem}

One can verify (see Lemma~\ref{lem:goesto2} in Appendix~\ref{sec:calculations}) that $e^{\gamma_q} < q$, meaning in particular that there is no tri-colored sum-free set in $\Z_q^n$ of size $q^{n - o(n)}$. When $q$ is not a prime power, one can also prove this, although the upper bound is not known to be as strong:

\begin{theorem}[\cite{blasiak}]
Let $q \geq 2$ be any positive integer, and let $\kappa := \frac12 \log((2/3)2^{3/2}) \approx 0.02831$. Then, any tri-colored sum-free set in $\Z_q^n$ has size at most $q^{n(1 - \kappa/q + o(1))}$.
\end{theorem}

For notational simplicity in our main results in Section~\ref{sec:main}, define $\gamma_q := (1-\kappa/q)\log(q)$ when $q \geq 2$ is not a power of a prime.

\subsection{Proof Outline}

In Section~\ref{sec:prelims} we formally define all the notions related to tensors which are necessary for the rest of the paper, and in Section~\ref{sec:modqdegens} we give our simple rank expression for $T_q$ and straightforward monomial degenerations of $T_{q+2}$ into $CW_q$ as well as other tensors $\T$ from past work on matrix multiplication algorithms. The remainder of the paper gives the proof of Theorem~\ref{thm:maininformal}, which proceeds in four main steps:
\begin{itemize}
\item In Section~\ref{sec:modqdegens}, we give a simple rank expression for $T_q$, and show that the rotated Coppersmith-Winograd tensor can be found as a simple monomial degeneration of $T_q$.
\item In Section~\ref{sec:triples}, we show that every matrix multiplication tensor has a zeroing out into a large number of independent triples. This generalizes a classical result that matrix multiplication tensors have monomial degenerations into a large number of independent triples.
\item In Section~\ref{sec:zeroings}, we show that if tensor $A$ is a monomial degeneration of tensor $B$, and large powers of $A$ can be zeroed out into many independent triples, then large powers of $B$ can as well.
\item Finally, in Section~\ref{sec:main}, we combine the above to show that if any tensor $\T$ is a monomial degeneration of $T_q$, and yields a fast matrix multiplication algorithm (meaning it can be zeroed out into many independent triples), then $T_q$ can be zeroed out into many independent triples as well. By noticing that independent triples in $T_q$ correspond to tri-colored sum-free sets, and combining with the upper bounds on the size of such a set, we get our lower bound.
\end{itemize}

\subsection{Comparison with Past Work}

There are two papers which have proved lower bounds on the value of $\omega$ that one can achieve using certain techniques.

The first is a work by Ambainis et al. \cite{ambainis}. They show a lower bound of $\Omega(n^{2.3078})$ for any algorithm for $n \times n \times n$ matrix multiplication one can design using the `laser method with merging' using the Coppersmith-Winograd tensor and its relatives. The laser method is a technique proposed by Strassen~\cite{laser} and used by all recent work \cite{coppersmith,stothers,v12,legall,legallrect2} in order to show that the Coppersmith-Winograd tensor has a zeroing out into many big disjoint matrix multiplication tensors (the second property of the two properties of a tensor $\T$ we described earlier). While the bound that Ambainis et al. get is better than ours, our result is much more general: First, the Ambainis et al. bound is for algorithms which use the Coppersmith-Winograd tensor and some tensors like it, whereas ours applies to any tensor which is an arbitrary monomial degeneration of $T_q$. Second, their bound only applies when the laser method with merging is used to zero out the tensor into matrix multiplication tensors, whereas ours applies to any possible monomial degeneration into matrix multiplication tensors.

The second prior work is by Blasiak et al. \cite{blasiak}. Like us, the authors also use recent bounds on the size of certain tri-colored sum-free sets in order to prove lower bounds. However, rather than the tensor-based approach to matrix multiplication algorithms which we have been discussing, and which has been used in all of the improvements to $\omega$ and $\alpha$ to date, they instead focus on the `group-theoretic approach' to matrix multiplication \cite{cohn2003group, cohn2005group}. This approach has been designed around formulating approaches that would imply $\omega = 2$ rather than on attempting any small improvement to the bounds on $\omega$, and this paper refutes some earlier conjectures along these lines. The work of Blasiak et al. implies that certain approaches to achieving $\omega=2$ are impossible, similar to our work here. 

In personal communication, Cohn~\cite{cohn-personal} stated that the Coppersmith-Winograd tensor CW$_q$ leads to a STPP (simultaneous triple product property) construction in $\Z_m^n$ with $m=q$ and $n$ tending to infinity. Blasiak et al. present lower bounds on what can be proved about $\omega$ using the group theoretic approach using STPP constructions in $\Z_m^n$ for any fixed $m$, and hence their results imply that the {\em group-theoretic variant} of the Coppersmith-Winograd approach cannot yield $\omega=2$ using a fixed $q$. It is not clear exactly what lower bounds this result implies for the original laser method approach, or for arbitrary monomial degenerations of $T_q$. Thus, we consider our results complementary to those of Blasiak et al.
Furthermore, our results include limitations for rectangular matrix multiplication, which the prior work does not mention.

\section{Preliminaries} \label{sec:prelims}

In this section we introduce all the notions related to tensors which are used in the rest of the paper.

\subsection{Tensor Definitions}

Let $X = \{x_1, \ldots, x_n \}$, $Y = \{y_1, \ldots, y_m \}$, and $Z = \{ z_1, \ldots, z_p \}$ be three sets of formal variables. A \emph{tensor over $X,Y,Z$} is a trilinear form $$T = \sum_{i=1}^n \sum_{j=1}^m \sum_{k=1}^p T_{ijk} x_i y_j z_k,$$ where the $T_{ijk}$ terms are elements of a field $\F$. The \emph{size} of a tensor $A$, denoted $|A|$, is the number of nonzero coefficients $A_{ijk}$. There are three particular tensors we will focus on in this paper. The \emph{matrix multiplication tensor $\langle n,m,p \rangle$} is given by $$\langle n,m,p \rangle = \sum_{i=1}^n \sum_{j=1}^m \sum_{k=1}^p  x_{ij} y_{jk} z_{ki}.$$ For a positive integer $q$, the \emph{structural tensor of $\Z_q$}, denoted $T_q$, is given by $$T_q = \sum_{i=0}^{q-1}\sum_{j=0}^{q-1} x_i y_j z_{-i-j\Mod{q}}.$$ 

For any positive integer $q$, the $q$th Coppersmith-Winograd tensor C$_q$~\cite{coppersmith} is given by $x_0 y_0 z_{q+1} + x_0 y_{q+1} z_0 + x_{q+1} y_0 z_0 +\sum_{i=1}^q (x_0y_iz_i + x_iy_0z_i+x_iy_iz_0)$. It is not hard to verify that using the Coppersmith-Winograd approach, one can obtain exactly the same values for $\omega$ from the following
 \emph{rotated Coppersmith-Winograd tensor} $CW_q$, given by
$$CW_q = x_0y_0z_{q+1} + x_0y_{q+1}z_0+x_{q+1}y_0z_0 + \sum_{k=1}^q \left( x_0 y_k z_{q+1-k} + x_k y_0 z_{q+1-k} + x_k y_{q+1-k} z_0 \right).$$

The main reason why CW$_q$ works just as well as the original Coppersmith-Winograd tensor C$_q$ is because they both have border rank $q+2$ and because of the following other structural reason which is what is used in the prior work on fast matrix multiplication:

Let $X_0 = \{x_0\}$, $X_1=\{x_1,\ldots, x_q\}$, $X_2=\{x_{q+2}\}$. Similarly, let $Y_0 = \{y_0\}$, $Y_1=\{y_1,\ldots, y_q\}$, $Y_2=\{y_{q+2}\}$, and $Z_0 = \{z_0\}$, $Z_1=\{z_1,\ldots, z_q\}$, $Z_2=\{z_{q+2}\}$. When you restrict C$_q$ and CW$_q$ to $X_0\times Y_2\times Z_0$, or $X_2\times Y_0\times Z_0$, or $X_0\times Y_0\times Z_2$, both of them are isomorphic to $\langle 1,1,1\rangle$. When you restrict them to $X_0\times Y_1\times Z_1$, both are isomorphic to $\langle 1,1,q\rangle$, when you restrict them to $X_1\times Y_0\times Z_1$, both are isomorphic to $\langle q,1,1\rangle$, and when you restrict them to $X_1\times Y_1\times Z_0$, both are isomorphic to $\langle 1,q,1\rangle$. 
The Coppersmith-Winograd approach only looks at products of these blocks in higher tensor powers, which are hence isomorphic to the same matrix multiplication tensors and give the same bounds on $\omega$.

\subsection{Subsets and Degenerations}

For two tensors $A,B$, we say that $A \subseteq B$ if $A_{ijk}$ is always either $B_{ijk}$ or $0$. For instance, we can see that $CW_q \subseteq T_{q+2}$. We furthermore say that $A$ is a \emph{monomial degeneration} of $B$ if $A \subseteq B$ and there are functions $a : X \to \Z$, $b : Y \to \Z$, and $c : Z \to \Z$ such that whenever $B_{ijk} \neq 0$,
\begin{itemize}
    \item we have $a(x_i) + b(y_j) + c(z_k) \geq 0$, and
    \item furthermore, $a(x_i) + b(y_j) + c(z_k) = 0$ if and only if $A_{ijk} \neq 0$ as well.
\end{itemize}

We note that in prior work, degenerations are defined via polynomials in a variable $\eps$, however when the degenerations are single monomials, the above definition is equivalent, where $a,b,c$ give the corresponding exponents of $\eps$.

Finally, we say that $A$ is a \emph{zeroing out} of $B$ if $A$ is a monomial degeneration of $B$ such that $a(x) \geq 0$ for all $x \in X$, $b(y) \geq 0$ for all $y \in Y$, and $c(z) \geq 0$ for all $z \in Z$. One can think of this as substituting $0$ for any variable which $a,b$, or $c$ maps to a positive value.

\subsection{Tensor Product}

Let $X,X',Y,Y',Z,Z'$ be sets of formal variables. If $A$ is a tensor over $X,Y,Z$, and $B$ is a tensor over $X',Y',Z'$, then the \emph{tensor product of $A$ and $B$}, denoted $A \otimes B$, is a tensor over $X \times X', Y \times Y', Z \times Z'$ given by
$$A \otimes B = \sum_{\substack{(x_i,x'_{i'}) \in X \times X' \\ (y_j, y'_{j'}) \in Y \times Y' \\ (z_k, z'_{k'}) \in  Z \times Z'}} A_{ijk} B_{i'j'k'} (x_i, x'_{i'}) (y_j, y'_{j'}) (z_k, z'_{k'}).$$
The $n$th tensor power of a tensor $A$, denoted $A^{\otimes n}$, is the result of tensoring $n$ copies of $A$ together. In other words, $A^{\otimes 1}= A$, and $A^{\otimes n} = A \otimes A^{\otimes n-1}$.

Tensor products preserve many key properties of tensors. For instance, if $A \subseteq C$ and $B \subseteq D$, then $A \otimes B \subseteq C \otimes D$, and this is also true if subset is replaced by monomial degeneration, or by zeroing out.

For a nonnegative integer $k$, if $A$ is a tensor over $X,Y,Z$, and if $X_1, \ldots, X_k$ are $k$ disjoint copies of $X$, and similar for $Y$ and $Z$, then $k \odot A$ denotes the (disjoint) sum of $k$ copies of $A$, one over $X_i, Y_i, Z_i$ for each $1 \leq i \leq k$.

\subsection{Independent Triples}

Two triples $(x,y,z), (x',y',z') \in X \times Y \times Z$ are \emph{independent} if $x\neq x'$, $y\neq y'$, and $z\neq z'$. A tensor $A$ is independent if, whenever $A_{ijk} \neq 0$ and $A_{i'j'k'} \neq 0$, and $(i,j,k) \neq (i',j',k')$, then the triples $(x_i,y_j,z_k)$ and $(x_{i'},y_{j'},z_{k'})$ are independent.

\subsection{Tensor Rank} A tensor $T$ over $X,Y,Z$ is a \emph{rank-one tensor} if there are coefficients $a_x$ for each $x \in X$, $b_y$ for each $y \in Y$, and $c_z$ for each $z \in Z$ in the underlying field $\F$ such that
$$T = \left( \sum_{x \in X} a_x \cdot x \right) \left( \sum_{y \in Y} b_y \cdot y \right) \left( \sum_{z \in Z} c_z \cdot z \right) = \sum_{(x,y,z) \in X \times Y \times Z} a_x b_y c_z \cdot xyz.$$
More generally, $T$ is a \emph{rank-$k$ tensor} if it can be written as the sum of $k$ rank-one tensors. The rank of $T$, denoted $R(T)$, is the smallest $k$ such that $R$ is a rank-$k$ tensor.

We can generalize this notion slightly to define the border rank of a tensor. We will now allow the $a_x$, $b_y$, and $c_z$ coefficients to be elements of the polynomial ring $\F[\eps]$ for a formal variable $\eps$. We say that $T$ is a \emph{border rank-one tensor} if there are coefficients $a_x, b_y, c_z$ in $\F[\eps]$ and an integer $h \geq 0$ such that when
\begin{align}\label{borderrank}\left( \sum_{x \in X} a_x \cdot x \right) \left( \sum_{y \in Y} b_y \cdot y \right) \left( \sum_{z \in Z} c_z \cdot z \right)\end{align}
is expanded as a polynomial in $\eps$ whose coefficients are tensors over $X,Y,Z$, then $T$ is the coefficient of $\eps^h$, and the coefficient of $\eps^{h'}$ is $0$ for all $0 \leq h' < h$. Similarly, the border rank $\underline{R}(T)$ of $T$ is the smallest number of expressions of the form (\ref{borderrank}) whose sum, when written as a polynomial in $\eps$, has $T$ as its lowest order coefficient.

It is not hard to see that if $A$ is a monomial degeneration of $B$, then $\underline{R}(B) \leq \underline{R}(A) \leq R(A)$.

\subsection{Matrix Multiplication Tensor and Algorithms}
Now that we have defined tensor rank, we can define $\omega$ as the infimum over all reals so that $R(\langle n,n,n\rangle)\leq O(n^{\omega+\eps})$ for all $\eps>0$. 
Similarly, for any $\eps\in (0,1)$, define $\omega_\eps$ to be the smallest real such that an $n\times n^\eps$ matrix can be multiplied by an $n^\eps\times n$ matrix in $n^{\omega_\eps+o(1)}$ time.

We present a useful Lemma that follows from the work of Sch\"onhage, which shows how the tensor rank notions we have been discussing can give bounds on $\omega_\eps$.

\begin{lemma} \label{lem:schon}
If $R(f~\odot \langle n,n^\eps,n \rangle)\leq g$, then $\omega_\eps \leq \log_n (\lceil g/f\rceil)$.
\end{lemma}

\begin{proof}
By Sch\"onhage \cite{sch81} (see also \cite[Lemma 7.7]{blaser}), we have that $R(f~\odot \langle n,n^\eps,n \rangle)\leq g$ implies that for all integers $s\geq 1$, $R(f~\odot \langle n^s,n^{s\eps},n^s\rangle )\leq f \lceil g/f\rceil^s$.
Hence, multiplying an $n^s\times (n^s)^\eps$ by an $(n/s)^\eps\times n^s$ matrix can be done in $O(f\lceil g/f\rceil^s)$ time.
Thus $\omega_\eps \leq \lim_{s\rightarrow \infty} \log (f\lceil g/f\rceil^s)/\log (n^s) = \log_n (\lceil g/f\rceil)$.
\end{proof}

We can also define $\alpha$ as the largest real such that $R(\langle n,n^\alpha, n\rangle)\leq n^{2+o(1)}$. It is known that $\alpha\in [0.31, 1]$, and clearly $\alpha=1$ if and only if $\omega=2$.

\section{The mod-p tensor and its degenerations} \label{sec:modqdegens}
In this section, we give a rank expression for $T_p$, and then a monomial degeneration of $T_{q+2}$ into $CW_q$. 

\subsection{The rank of $T_p$}
Let us consider the tensor $T_p$ of addition modulo $p$ for any integer $p\geq 2$; recall that in trilinear notation, $T_p$ is defined as $$T_p = \sum_{\substack{i,j,k\in\{0,\ldots,p-1\} \\ i+j+k\equiv 0\Mod{p}}} x_iy_jz_k.$$

The rank of $T_p$ is $p$, as can be seen by the expression below. Let $w_1,\ldots,w_p\in \mathbb{C}$ be the $p$th roots of unity, meaning that $\sum_{i=1}^p w_i = 0$, and that for each $i$, $w_i^p = 1$. Then,

$$T_p = \frac{1}{p} \sum_{\ell=1}^p \left(\sum_{i=0}^{p-1} w_{\ell}^{i} x_i\right)\left(\sum_{j=0}^{p-1} w^j_{\ell}y_j\right)\left(\sum_{k=0}^{p-1} w^k_{\ell} z_k\right).$$

The above gives a rank expression for $T_q$ over $\mathbb{C}$, which is sufficient for the approaches for matrix multiplication algorithms discussed above. That said, one can easily modify it to get an expression over some other fields as well. For instance, suppose $p+1$ is an odd prime. Then, we know that $\sum_{a=1}^{p} a \equiv 0\Mod{p+1}$, and that $a^{p}\equiv 1 \Mod{p+1}$ for any $1 \leq a \leq p$, so we similarly get the following rank expression over $GF(p+1)$:

$$T_{p} = - \sum_{a=1}^p \left(\sum_{i=0}^{p-1} a^{i} x_i\right)\left(\sum_{j=0}^{p-1} a^j y_j\right)\left(\sum_{k=0}^{p-1} a^k z_k\right).$$

\subsection{Monomial degeneration of $T_{q+2}$ into $CW_q$}

Here we will show that the rotated CW tensor $CW_{q}$ for integer $q\geq 1$ is a degeneration of $T_{q+2}$. Recall that
\begin{align}\label{CWrotated}CW_q = x_0y_0z_{q+1} + x_0y_{q+1}z_0+x_{q+1}y_0z_0 + \sum_{k=1}^q \left( x_0 y_k z_{q+1-k} + x_k y_0 z_{q+1-k} + x_k y_{q+1-k} z_0 \right).\end{align}
For ease of notation, we will change the indexing of the $z$ variables in $T_{q+2}$ (i.e. rename the variables) from our original definition\footnote{For every index $k\in \{0,1,\ldots,q+1\}$, we will rename $z_k$ to $z_{k-1\pmod{q+2}}$.} to write
\begin{align}\label{Trelabeled}T_{q+2} = \sum_{\substack{i,j,k\in\{0,\ldots,q+1\} \\ i+j+k\equiv q+1\Mod{q+2}}} x_iy_jz_k.\end{align}

In this form, one can see that $CW_q$ is the subset of $T_{q+2}$ consisting of all the terms containing at least one of $x_0$, $y_0$, or $z_0$. With this in mind, our degeneration of $T_{q+2}$ is as follows. We will pick:
\begin{itemize}
    \item $a(x_0)=0$, $a(x_{q+1})=2$, and $a(x_i) = 1$ for $1 \leq i \leq q$, similarly,
    \item $b(y_0)=0$, $b(y_{q+1})=2$, and $b(y_j) = 1$ for $1 \leq j \leq q$, and,
    \item $c(z_0)=-2$, $c(z_{q+1})=0$, and $c(z_k) = -1$ for $1 \leq k \leq q$.
\end{itemize}

We need to verify that for every term $x_i y_j z_k$ in (\ref{Trelabeled}) we have $a(x_i) + b(y_j) + c(z_k) \geq 0$, and moreover that for such $x_i y_j z_k$, $a(x_i) + b(y_j) + c(z_k) = 0$ if and only if $x_i y_j z_k$ also appears in (\ref{CWrotated}). This is quite straightforward, but we do it here for completeness. Consider any term $x_i y_j z_k$ in (\ref{Trelabeled}). We consider three cases based on $k$:
\begin{itemize}
    \item If $k=0$, then our term is of the form $x_i y_{q+2-i} z_0$ for $0 \leq i \leq q+2$. This term always appears in (\ref{CWrotated}) as well, and we can see that we always have $a(x_i) = 2 - b(y_{q+2-i})$, and so $a(x_i) + b(y_{q+2-i}) + c(z_0) = 0$.
    \item If $k=q+1$, then $c(z_{q+1}) = 0$, and we always have $a,b \geq 0$, so we definitely have that $a(x_i) + b(y_j) + c(z_k) \geq 0$. Moreover, we can only achieve $0$ when $a=b=0$, with the term $x_0 y_0 z_{q+1}$, which is the only term with $z_{q+1}$ which appears in (\ref{CWrotated}).
    \item If $1 \leq k \leq q$, then since $x_0 y_0 z_k$ is not a term in (\ref{Trelabeled}), we must have that $a(x_i) + b(y_j) \geq 1$, and so $a(x_i) + b(y_j) + c(z_k) \geq 0$. Moreover, we only achieve $a(x_i) + b(y_j) + c(z_k) = 0$ when $(a,b) = (0,1)$ or $(1,0)$, which correspond to the terms of the form $ x_0 y_k z_{q+1-k}$ or $x_k y_0 z_{q+1-k}$ in (\ref{CWrotated}).
\end{itemize}

\subsection{Monomial degeneration of $T_{q+1}$ into Strassen's 1986 tensor.}

Strassen's 1986 tensor is defined for any integer $q\geq 1$ and is given by $S_q:=\sum_{i=1}^q x_0 y_i z_{q+1-i} + x_i y_0 z_{q+1-i}$. 

Similar to before, we will show that $S_q$ is a degeneration of $T_{q+1}$, which we can write as
\begin{align}\label{Trelabeled2}T_{q+1} = \sum_{\substack{i,j,k\in\{0,\ldots,q\} \\ i+j+k\equiv q\Mod{q+1}}} x_iy_jz_k.\end{align}

Our degeneration is as follows: $a(x_0) = b(x_0) = 0$, $a(x_i)=b(y_i)=1$ for all $i\geq 1$, $c(z_q)=0$ and $c(z_k)=-1$ for all $k\geq 1$. Simple casework shows again that the possible values for $a(x_i)+b(y_j)+c(z_{k})$ are $0,1,2$, and that $0$ is only achieved for the terms in $S_q$.
Among other things, this degeneration gives a simple proof that the border rank of $S_q$ is $q+1$.

Since a monomial degeneration of a rank expression gives a border rank expression, this shows in particular that the border rank of CW$_{q}$ is $q+2$. Furthermore, it shows that the best known bounds for $\omega$~\cite{coppersmith,v12,legall} can be obtained from $T_7$. Finally, since we only used monomial degenerations, we will be able to obtain lower bounds on what bounds on $\omega$ one can achieve via zeroing out powers of the CW$_{q}$ tensor.

\section{Independent Triples in Matrix Multiplication Tensors} \label{sec:triples}

In this section we show that there is a zeroing out of any matrix multiplication tensor into a fairly large independent tensor. This strengthens a classic result (see eg. \cite[Lemma 8.6]{blaser}) that any matrix multiplication tensor has a monomial degeneration into a fairly large independent tensor.

\begin{lemma}\label{lemma:indep}
For every positive integer $q$, and $\eps \in (0,1]$, there is a zeroing out of $\langle q, q^{\eps}, q \rangle^{\otimes n}$ into $q^{(1+\eps)n - o(n)}$ independent triples.
\end{lemma}

\begin{proof}
Recall that $\langle q, q^\eps, q \rangle = \sum_{i=1}^q \sum_{j=1}^{q^\eps} \sum_{k=1}^q x_{ij} y_{jk} z_{ki}$. Hence,
$$\langle q, q^\eps, q \rangle^{\otimes n} = \sum_{\i,\k \in [q]^n,~ \j \in [q^\eps]^n} x_{\i\j} y_{\j\k} z_{\k\i}.$$
We will zero out variables in three phases, and after the third phase we will have a sufficiently large independent tensor as desired.

\subsection{Phase one} For vectors  $\i,\k \in [q]^n$, and values $a,b \in [q]$, let $t_{ab}(\i\k)$ denote the number of $1 \leq \alpha \leq n$ such that $\i_\alpha = a$ and $\k_\alpha = b$. We say that $\i\k$ is \emph{balanced} if, for all $a,b,c,d \in [q]$, we have $t_{ab}(\i\k)=t_{cd}(\i\k)$. We similarly say that $\i\j$ is balanced if $t_{ab}(\i\j) = t_{cd}(\i\j)$ for every $a,c \in [q]$ and $b,d \in [q^\eps]$, and say that $\j\k$ is balanced similarly. In the first phase, we zero out every variable $x_{\i\j}$ such that $\i\j$ is not balanced. We similarly zero out $y_{\j\k}$ such that $\j\k$ is not balanced, and $z_{\k\i}$ such that $\k\i$ is not balanced.

Note that if $\i\k$ is balanced, then for each $a,b\in[q]$, we have $(\i_\alpha,\k_\alpha)=(a,b)$ for exactly $n/q^2$ choices of $\alpha \in [n]$. Hence, the number of choices of $\i,\k \in [q]^n$ such that $\i\k$ is balanced is exactly $L_2 := \binom{n}{\frac{n}{q^2}, \frac{n}{q^2}, \ldots, \frac{n}{q^2}} = q^{2n - o(n)}$. If $\i\k$ is balanced, then notice that the number $K_\eps$ of choices of $\j \in [q^\eps]^n$ such that $\i\j$ and $\j\k$ are also balanced is independent of what $\i$ and $\k$ are, and satisfies $K_\eps = q^{O(n)}$.

Similarly, the number of choices of $\i \in [q]^n$ and $\j \in [q^\eps]^n$ such that $\i\j$ is balanced is $L_{1+\eps} := \binom{n}{\frac{n}{q^{1+\eps}}, \frac{n}{q^{1+\eps}}, \ldots, \frac{n}{q^{1+\eps}}} = q^{(1+\eps)n - o(n)}$. Moreover, when $\i\j$ is balanced, the number $K_1$ of choices of $\k$ such that $\i\k$ and $\j\k$ are balanced satisfies $K_1 = q^{O(n)}$. Note that $L_2 K_\eps = L_{1+\eps} K_1$, since both count the number of triples remaining after phase one, and in particular, $K_1 \geq K_\eps$.

\subsection{Phase two} Let $M$ be an odd prime number to be determined. Pick $w_0, w_1, \ldots, w_n \in [M]$ independently and uniformly at random, then define the hash functions $h_X : X \to [M]$, $h_Y : Y \to [M]$, and $h_Z : Z \to [M]$, by:

$$h_X(x_{\i\j}) = 2\sum_{\alpha = 1}^n w_\alpha \cdot (\i_\alpha - \j_\alpha) \pmod{M},$$
$$h_Y(y_{\j\k}) = 2 w_0 + 2\sum_{\alpha = 1}^n w_\alpha \cdot (\j_\alpha - \k_\alpha) \pmod{M},$$
$$h_Z(z_{\k\i}) = w_0 + \sum_{\alpha = 1}^n w_\alpha \cdot (\i_\alpha - \k_\alpha) \pmod{M}.$$

Notice that, for every choice of $\i,\j,\k \in [q]^n$, we have that $h_X(x_{\i\j}) + h_Y(y_{\j\k}) = 2h_Z(z_{\k\i}) \Mod{M}$. Now, let $H \subseteq [M]$ be a subset of size $|H| \geq M^{1-o(1)}$ which does not contain any nontrivial three-term arithmetic progressions mod $M$; in other words, if $a,b,c \in H$ such that $a+b=2c \Mod{M}$, then $a=b=c$. Such a set is constructed by Salem and Spencer~\cite{salemspencer}. In the second phase, we zero out all $x_{\i\j}$ such that $h_X(x_{\i\j}) \notin H$, and similarly for the $y$ and $z$ variables. As a result, every term $x_{\i\j}y_{\j\k}z_{\k\i}$ remaining in our tensor satisfies:
\begin{itemize}
    \item $\i\j$, $\j\k$, and $\k\i$ are balanced, and
    \item $h_X(x_{\i\j}) = h_Y(y_{\j\k}) = h_Z(z_{\k\i})$.
\end{itemize}

\subsection{Phase three} In the third phase we zero out some remaining variables to ensure that our resulting tensor is independent. First, however, we will compute some expected values.

For $h \in H$, let $S_h$ be the set of terms $x_{\i\j}y_{\j\k}z_{\k\i}$ remaining in our tensor after stage two such that $h_X(x_{\i\j}) = h_Y(y_{\j\k}) = h_Z(z_{\k\i}) = h$. For a given term $x_{\i\j}y_{\j\k}z_{\k\i}$ which was not zeroed out in phase one, it will be in $S_h$ whenever $h_X(x_{\i\j}) = h$ and $h_Y(y_{\j\k}) = h$, since in that case we must also have that $h_Z(z_{\k\i}) = h$ as the three are in arithmetic progression. For a fixed choice of $\i,\j,\k$ such that $\i\j$ and $\j\k$ are balanced, we can see that $h_X(x_{\i\j})$ and $h_Y(y_{\j\k})$ are independent and uniformly random elements of $[M]$ (the randomness is over choosing the $w_\alpha$ values). Hence, this term will be in $S_h$ with probability $1/M^2$, and so $\E[|S_h|] = L_{1+\eps} \cdot K_1 / M^2$.

Next, for $h \in H$, let $P_h$ be the set of pairs of terms $(x_{\i\j}y_{\j\k}z_{\k\i}, x_{\i'\j'}y_{\j'\k'}z_{\k'\i'})$ such that both terms are in $S_h$, and $\i=\i'$ and $\j=\j'$, meaning they share the same $x$ variable. Again, there are $L_{1+\eps}$ choices for $\i$ and $\j$, then $K_1$ choices each for $\k$ and $\k'$, and similar to before, such a choice of $\i,\j,\k,\k'$ will be put in $P_h$ with probability $1/M^3$. Hence, $\E[|P_h|] \leq L_{1+\eps} \cdot K_1^2 / M^3$. Similar calculations hold if we instead look at pairs $Q_h$ which share a $y$ variable, showing that $\E[|Q_h|] \leq L_{1+\eps} \cdot K_1^2 / M^3$, or pairs $R_h$ which share a $z$ variable, showing that $\E[|R_h|] \leq L_{2} \cdot K_\eps^2 / M^3 \leq L_{1+\eps} \cdot K_1^2 / M^3$.

We now do our final zeroing out. If there are any distinct terms $x_{\i\j} y_{\j\k} z_{\k\i}$ and $x_{\i'\j'} y_{\j'\k'} z_{\k'\i'}$ remaining in our tensor such that $\i=\i'$ and $\j = \j'$, then we zero out $x_{\i\j}$. We similarly zero out any variables $y_{\j\k}$ or $z_{\k\i}$ which appear in multiple terms. As a result, our final tensor is definitely independent.

It remains to show that it has enough terms remaining. Since each pair of terms left from phase two which share a variable is removed in phase three, we see that the number of terms remaining is at least

$$\sum_{h \in H} |S_h| - 2|P_h| - 2|Q_h| - 2|R_h|.$$
Let us pick $M$ to be an odd prime number in the range $[12K_1, 24K_1]$. Hence, using our expected value calculations from before, we see that the expected number of remaining terms is at least
\begin{align*}|H| \cdot \left( \frac{L_{1+\eps} K_1}{M^2} - 6 \frac{L_{1+\eps} K_1^2}{M^3} \right) = \frac{|H| L_{1+\eps} K_1}{M^2} \left( 1 - 6 \frac{K_1}{M} \right) &\geq \frac{M^{1-o(1)} L_{1+\eps} K_1}{M^2} \left( 1 - 6 \frac{1}{12} \right)\\ &\geq \frac{L_{1+\eps}}{K_1^{o(1)}} \geq q^{(1+\eps)n - o(n)},\end{align*}
where the last step follows since $L_{1+\eps} = q^{(1+\eps)n - o(n)}$ and $K_1 = q^{O(n)}$.
By the probabilistic method, there is a choice of hash functions which achieves this expected number of independent triples, as desired.
\end{proof}

\section{Monomial Degenerations} \label{sec:zeroings}

\begin{lemma}\label{lemma:mondeg}
Suppose $A$ and $B$ are two tensors over $X,Y,Z$ such that $A$ is a monomial degeneration of $B$. Further suppose that $A^{\otimes n}$ has zeroing out into $f(n)$ independent triples. Then, $B^{\otimes n}$ has a zeroing out into $\Omega(f(n)/n^2)$ independent triples.
\end{lemma}

\begin{proof}
Let $a : X \to \Z$, $b : Y \to \Z$, and $c : Z \to \Z$ be the functions for the monomial degeneration such that
\begin{itemize}
    \item $a(x_i) + b(y_j) + c(z_k) \geq 0$ for all $x_i y_j z_k \in B$, and
    \item furthermore $a(x_i) + b(y_j) + c(z_k) = 0$ if and only if $x_i y_j z_k \in A$.
\end{itemize}
Let $a^- := \min_{x \in X} a(x)$ and $a^+ := \max_{x \in X} a(x)$, and define $b^-$, $b^+$, $c^-$, and $c^+$ similarly. Now, $B^{\otimes n}$ is a tensor over $X^n, Y^n, Z^n$. Define $a^n : X^n \to \Z$, by $a^n(x_{i_1}, \ldots, x_{i_n}) = \sum_{\alpha=1}^n a(x_{i_\alpha})$, and define $b^n : Y^n \to \Z$ and $c^n : Z^n \to \Z$ similarly. It follows that
\begin{itemize}
    \item $a^n(x_{i_1}, \ldots, x_{i_n}) + b^n(y_{j_1},\ldots,y_{j_n}) + c^n(z_{k_1},\ldots,z_{k_n}) \geq 0$ for all $x_{i_1} \cdots x_{i_n} y_{j_1}\cdots y_{j_n}z_{k_1}\cdots z_{k_n} \in B^{\otimes n}$, and
    \item furthermore $a^n(x_{i_1}, \ldots, x_{i_n}) + b^n(y_{j_1},\ldots,y_{j_n}) + c^n(z_{k_1},\ldots,z_{k_n}) = 0$ if and only if \\ $x_{i_1} \cdots x_{i_n} y_{j_1}\cdots y_{j_n}z_{k_1}\cdots z_{k_n} \in A^{\otimes n}$.
\end{itemize}
The range of $a^n$ is integers in $[a^- n, a^+ n]$. For each integer $p$ in that range, let $X^n_p$ be the set of $x_{i_1} \cdots x_{i_n} \in X^n$ such that $a^n(x_{i_1} \cdots x_{i_n}) = p$. Define $Y^n_q$ for integers $q \in [b^- n, b^+ n]$, and $Z^n_r$ for integers $r \in [c^- n, c^+ n]$, similarly. Now, for $(p,q,r) \in [a^- n, a^+ n] \times [b^- n, b^+ n] \times [c^- n, c^+ n]$, let $B^{\otimes n}_{p,q,r}$ be the tensor one gets from $B^{\otimes n}$ by zeroing out all the $X^n$ variables not in $X^n_p$, all the $Y^n$ variables not in $Y^n_q$, and all the $Z^n$ variables not in $Z^n_r$. Then, letting $W$ be the set of triples of integers in $[a^- n, a^+ n] \times [b^- n, b^+ n] \times [c^- n, c^+ n]$, we see that
$$A^{\otimes n} = \sum_{(p,q,r) \in W \mid p+q+r=0} B^{\otimes n}_{p,q,r},$$
and each term of $A^{\otimes n}$ appears in exactly one of the summands. Now, let $A^{\otimes n \prime}$ be the zeroing out of $A^{\otimes n}$ into $f(n)$ independent triples. Let $B^{\otimes n \prime}_{p,q,r}$ be the zeroing out of $B^{\otimes n}_{p,q,r}$ in which we zero out those same variables. Hence,
$$A^{\otimes n \prime} = \sum_{(p,q,r) \in W \mid p+q+r=0} B^{\otimes n \prime}_{p,q,r},$$
where the sum is hence a disjoint sum of independent triples. The number of terms on the right is $O(n^2)$, and so at least one of the terms on the right must have size at least $|A^{\otimes n \prime}| / O(n^2) = \Omega(f(n) / n^2)$, as desired.
\end{proof}

\section{Main Theorem}\label{sec:main}

In this section, we will combine our results above with the bounds on the sizes of tri-colored sum-free sets from past work in order to prove our main theorem. Recall the definition of $\gamma_p$ from Section \ref{sec:trisets}, and define $c_p := e^{\gamma_p}$.

\begin{theorem}\label{thm:main}Let $\eps\in (0,1]$. Let T be a tensor that is a monomial degeneration of $T_p$ 
and suppose that $T^{\otimes N}$ can be zeroed out into $F~\odot \langle G,G^\eps,G\rangle $, giving a bound $\omega_\eps\leq \omega'_\eps$ where 
$G^{\omega'_\eps} = \lceil p^N/F\rceil$.
Then $\omega'_\eps\geq (1+\eps)\log_{c_p} p$. 
\end{theorem}

\begin{proof} Let $g=G^{1/N}$ so that $G=g^N$, and let $f=F^{1/N}$ so that $F=f^N$.
Since $T^{\otimes N}$ can be zeroed out into $F~\odot \langle G,G^\eps,G\rangle $, via Lemma~\ref{lemma:indep}, $T^{\otimes N}$ can be zeroed out into $f^N\cdot g^{(1+\eps)N-o(N)}$ independent triples.
Due to Lemma~\ref{lemma:mondeg} this means that $T_p^{\otimes N}$ can also be zeroed out into $D=f^N\cdot g^{(1+\eps)N-o(N)}/N^2$ independent triples.

Now, let $S=\{(a_1,b_1,c_1),\ldots,(a_D,b_D,c_D)\}$ be the 
 indices of the $D$ independent triples obtained from $T_p^{\otimes N}$. Because they are obtained by zeroing out $T_p^{\otimes N}$, for every $i$, $a_i+b_i+c_i\equiv 0$ in $Z_p^N$. Now suppose that for some $i,j,k$, $a_i+b_j+c_k\equiv 0$ in $Z_p^N$.
If $i,j,k$ are not all the same, then $(a_i,b_j,c_k)$ cannot be in $S$ as the triples in $S$ are independent. However, the only way for a triple of $T_p^{\times N}$ to be removed is if $X_{a_i}$ or $Y_{b_j}$ or $Z_{c_k}$ is set to zero. Suppose that $X_{a_i}$ is set to $0$ (the other two cases are symmetric). Then there can be no triple in $S$ sharing $a_i$ as its first index. Thus in fact $S$ forms a tri-colored sum-free set. Hence $D\leq c_p^{N}$.

From our earlier bound on $D$ we get that $f^N\cdot g^{(1+\eps)N-o(N)}/N^2 \leq c_p^N$, and taking the $N$th root of both sides yields $f g^{1+\eps-o(1)}/N^{2/N}\leq c_p$.

Recall that $G^{\omega'_\eps} = \lceil p^N/F\rceil$, so that $g=(\lceil p/f\rceil)^{1/\omega'_\eps}$. Plugging in above, we get that 
$f (\lceil p/f\rceil)^{(1+\eps)/\omega'_\eps -o(1)}\leq c_p.$ 
Hence, $f^{1-(1+\eps)/\omega'_\eps + o(1)} p^{(1+\eps)/\omega'_\eps-o(1)}\leq c_p.$
Since $\omega'_\eps\geq (1+\eps)$, we have that $f^{1-(1+\eps)/\omega'_\eps + o(1)}\geq 1$.
We obtain that
$(1+\eps)/\omega'_\eps \leq \log_p c_p + o(1)$ and 
\begin{align*}
\omega'_\eps\geq (1+\eps-o(1))\log_{c_p} p.
\end{align*}
\end{proof}

As a corollary we obtain the following upper bound on what $\alpha$ can be achieved by zeroing out.

\begin{corollary} Let $T$ be a tensor that is a monomial degeneration of $T_p$. If one can prove $\alpha\leq \alpha'$ using the zeroing-out approach then, $\alpha'\leq \frac{2}{\log_{c_p} p} -1$.
\end{corollary}

\bibliographystyle{alpha}
\bibliography{papers}

\newcommand{\etalchar}[1]{$^{#1}$}
\begin{thebibliography}{BCC{\etalchar{+}}17}

\bibitem[AFLG15]{ambainis}
Andris Ambainis, Yuval Filmus, and Fran{\c{c}}ois Le~Gall.
\newblock Fast matrix multiplication: limitations of the coppersmith-winograd
  method.
\newblock In {\em STOC}, pages 585--593, 2015.

\bibitem[BCC{\etalchar{+}}17]{blasiak}
Jonah Blasiak, Thomas Church, Henry Cohn, Joshua~A Grochow, Eric Naslund,
  William~F Sawin, and Chris Umans.
\newblock On cap sets and the group-theoretic approach to matrix
  multiplication.
\newblock {\em Discrete Analysis}, 2017(3):1--27, 2017.

\bibitem[BCRL79]{BCRL79}
D.~Bini, M.~Capovani, F.~Romani, and G.~Lotti.
\newblock ${O}(n^{2.7799})$ complexity for $n\times n$ approximate matrix
  multiplication.
\newblock {\em Inf. Process. Lett.}, 8(5):234--235, 1979.

\bibitem[Bl{\"a}13]{blaser}
Markus Bl{\"a}ser.
\newblock Fast matrix multiplication.
\newblock {\em Theory of Computing, Graduate Surveys}, 5:1--60, 2013.

\bibitem[CKSU05]{cohn2005group}
Henry Cohn, Robert Kleinberg, Balazs Szegedy, and Christopher Umans.
\newblock Group-theoretic algorithms for matrix multiplication.
\newblock In {\em FOCS}, pages 379--388, 2005.

\bibitem[Coh17]{cohn-personal}
Henry Cohn.
\newblock personal communication, 2017.

\bibitem[Cop97]{coppersmith1997rectangular}
Don Coppersmith.
\newblock Rectangular matrix multiplication revisited.
\newblock {\em Journal of Complexity}, 13(1):42--49, 1997.

\bibitem[CU03]{cohn2003group}
Henry Cohn and Christopher Umans.
\newblock A group-theoretic approach to fast matrix multiplication.
\newblock In {\em FOCS}, pages 438--449, 2003.

\bibitem[CW81]{cw81as}
D.~Coppersmith and S.~Winograd.
\newblock On the asymptotic complexity of matrix multiplication.
\newblock In {\em SFCS}, pages 82--90, 1981.

\bibitem[CW90]{coppersmith}
Don Coppersmith and Shmuel Winograd.
\newblock Matrix multiplication via arithmetic progressions.
\newblock {\em Journal of symbolic computation}, 9(3):251--280, 1990.

\bibitem[DS13]{stothers}
A.M. Davie and A.~J. Stothers.
\newblock Improved bound for complexity of matrix multiplication.
\newblock {\em Proceedings of the Royal Society of Edinburgh, Section: A
  Mathematics}, 143:351--369, 4 2013.

\bibitem[EG17]{eg}
Jordan~S Ellenberg and Dion Gijswijt.
\newblock On large subsets of $\mathbb{F}_q^n$ with no three-term arithmetic
  progression.
\newblock {\em Annals of Mathematics}, 185(1):339--343, 2017.

\bibitem[GU17]{legallrect2}
Fran{\c{c}}ois~Le Gall and Florent Urrutia.
\newblock Improved rectangular matrix multiplication using powers of the
  coppersmith-winograd tensor.
\newblock {\em arXiv preprint arXiv:1708.05622}, 2017.

\bibitem[KSS16]{kleinberg}
Robert Kleinberg, William~F Sawin, and David~E Speyer.
\newblock The growth rate of tri-colored sum-free sets.
\newblock {\em arXiv preprint arXiv:1607.00047}, 2016.

\bibitem[LG12]{legallrect}
Fran{\c{c}}ois Le~Gall.
\newblock Faster algorithms for rectangular matrix multiplication.
\newblock In {\em FOCS}, pages 514--523, 2012.

\bibitem[LG14]{legall}
Fran{\c{c}}ois Le~Gall.
\newblock Powers of tensors and fast matrix multiplication.
\newblock In {\em ISSAC}, pages 296--303, 2014.

\bibitem[Mic14]{michalek-personal}
Mateusz Michalek.
\newblock personal communication, 2014.

\bibitem[Nor16]{norin2016}
Sergey Norin.
\newblock A distribution on triples with maximum entropy marginal.
\newblock {\em arXiv preprint arXiv:1608.00243}, 2016.

\bibitem[Pan78]{Pan78}
V.~Y. Pan.
\newblock Strassen's algorithm is not optimal.
\newblock In {\em FOCS}, volume~19, pages 166--176, 1978.

\bibitem[Pan80]{Pan80}
V.~Y. Pan.
\newblock New fast algorithms for matrix operations.
\newblock {\em SIAM J. Comput.}, 9(2):321--342, 1980.

\bibitem[Peb16]{pebody}
Luke Pebody.
\newblock Proof of a conjecture of kleinberg-sawin-speyer.
\newblock {\em arXiv preprint arXiv:1608.05740}, 2016.

\bibitem[Sch81]{sch81}
A.~Sch\"{o}nhage.
\newblock Partial and total matrix multiplication.
\newblock {\em SIAM J. Comput.}, 10(3):434--455, 1981.

\bibitem[SS42]{salemspencer}
Rapha{\"e}l Salem and Donald~C Spencer.
\newblock On sets of integers which contain no three terms in arithmetical
  progression.
\newblock {\em Proceedings of the National Academy of Sciences},
  28(12):561--563, 1942.

\bibitem[Str69]{strassen}
Volker Strassen.
\newblock Gaussian elimination is not optimal.
\newblock {\em Numerische mathematik}, 13(4):354--356, 1969.

\bibitem[Str86]{laser}
V.~Strassen.
\newblock The asymptotic spectrum of tensors and the exponent of matrix
  multiplication.
\newblock In {\em FOCS}, pages 49--54, 1986.

\bibitem[Str87]{strassenlaser1}
V.~Strassen.
\newblock Relative bilinear complexity and matrix multiplication.
\newblock {\em J. reine angew. Math. (Crelles Journal)}, 375--376:406--443,
  1987.

\bibitem[Wil12]{v12}
Virginia~Vassilevska Williams.
\newblock Multiplying matrices faster than coppersmith-winograd.
\newblock In {\em STOC}, pages 887--898, 2012.

\end{thebibliography}

\appendix

\section{Supporting Calculations} \label{sec:calculations}

We recall some definitions from earlier in the paper.
For any integer $q \geq 2$, let $\rho$ be the unique number in $(0,1)$ satisfying
$$\rho + \rho^2 + \cdots + \rho^{q-1} = \frac{q-1}{3} (1 + 2 \rho^q ).$$
Then, define $\gamma_q \in \R$ by $\gamma_q := \ln(1-\rho^q) - \ln(1-\rho) - \frac{q-1}{3} \ln(\rho)$. Then, the lower bound on $\omega$ we get from using $T_q$ is $2 \ln(q)/\gamma_q$. Here we show that this approaches $2$ as $q \to \infty$:

\begin{lemma} \label{lem:goesto2}
$\lim_{q \to \infty} \frac{\gamma_q}{\ln(q)} = 1$.
\end{lemma}

\begin{proof}
Note that, since $\rho \in (0,1)$, we have
$$\frac{1}{1-\rho} = 1 + \rho + \rho^2 + \cdots > \rho + \rho^2 + \cdots + \rho^{q-1} = \frac{q-1}{3} (1 + 2 \rho^q ) > \frac{q-1}{3} .$$
Rearranging, we see that $\rho > 1 - 3/(q-1)$. Hence,
\begin{align*}\frac{\gamma_q}{\ln(q)} = \frac{\ln \left( \frac{1 - \rho^q}{1 - \rho}  \right)}{\ln(q)} + \frac{(q-1) \ln(\rho)}{3\ln(q)} &> \frac{\ln \left( 1 + \rho + \cdots + \rho^{q-1}  \right)}{\ln(q)} + \frac{(q-1) \ln(1 - \frac{3}{q-1})}{3\ln(q)} \\ &>  \frac{\ln \left( (q-1)/3  \right)}{\ln(q)} + \frac{(q-1) \ln(1 - \frac{3}{q-1})}{3\ln(q)}.\end{align*}

As $q \to \infty$, we have that $\ln_q((q-1)/3) \to 1$ and $(q-1) \ln_q (1 - 3/(q-1)) \to 0$, as desired.
\end{proof}

\end{document}